\documentclass{scrartcl}
\KOMAoptions{paper=letter} 

\PassOptionsToPackage{hyphens}{url}

\RequirePackage{fix-cm}
\RequirePackage[T1]{fontenc}
\RequirePackage{xspace}
\RequirePackage{amssymb}
\RequirePackage{amsmath}
\RequirePackage{amsthm}
\RequirePackage[hidelinks]{hyperref} 
\RequirePackage[round]{natbib}

\RequirePackage{nicefrac}
\RequirePackage{tikz}
\RequirePackage{pgflibraryshapes}
\usetikzlibrary{positioning,arrows,shapes,snakes,calc,fit}
\RequirePackage{colortbl}
\RequirePackage{booktabs}
\RequirePackage{cleveref}
\RequirePackage{array}
\RequirePackage{graphicx}
\RequirePackage{xcolor}
\RequirePackage{enumitem}
\RequirePackage{tabularx}

\date{} 

\theoremstyle{definition}
\newtheorem{definition}{Definition}
\newtheorem{example}{Example}
\theoremstyle{plain}
\newtheorem{theorem}{Theorem}
\newtheorem{lemma}{Lemma}
\newtheorem{proposition}{Proposition}

\newtheoremstyle{bfnote}
{}{}%
{}{}%
{\bfseries}{.}%
{ }%
{\thmname{#1}\thmnumber{ #2}\thmnote{ (#3)}}
\theoremstyle{bfnote}

\newcommand{\supp}{\mathrm{supp}} 

\newcommand{\Omit}[1]{} 

\usepackage{multirow}
\usepackage{thm-restate}

\theoremstyle{remark}

\newtheorem{step}{Step}

\DeclareMathOperator*{\argmax}{arg\,max}
\DeclareMathOperator*{\argmin}{arg\,min}

\newcommand{\norm}[1]{\| #1 \|}
\newcommand{\prof}{C}
\newcommand{\csum}{|C|}
\newcommand{\cnom}{C} 
\newcommand{\nash}[1][]{\ifthenelse{\equal{#1}{}}{\ensuremath{\mathit{NASH}}\xspace}{\ensuremath{\mathit{NASH}(#1)}\xspace}}
\newcommand{\util}[1][]{\ifthenelse{\equal{#1}{}}{\ensuremath{\mathit{UTIL}}\xspace}{\ensuremath{\mathit{UTIL}(#1)}\xspace}}
\newcommand{\cut}[1][]{\ifthenelse{\equal{#1}{}}{\ensuremath{\mathit{CUT}}\xspace}{\ensuremath{\mathit{CUT}(#1)}\xspace}}
\newcommand{\anticut}[1][]{\ifthenelse{\equal{#1}{}}{\ensuremath{\mathit{ANTICUT}}\xspace}{\ensuremath{\mathit{ANTICUT}(#1)}\xspace}}
\newcommand{\unc}[1][]{\ifthenelse{\equal{#1}{}}{\ensuremath{\mathit{UNC}}\xspace}{\ensuremath{\mathit{UNC}(#1)}\xspace}}
\newcommand{\ceg}[1][]{\ifthenelse{\equal{#1}{}}{\ensuremath{\mathit{CEG}}\xspace}{\ensuremath{\mathit{CEG}(#1)}\xspace}}
\newcommand{\egal}[1][]{\ifthenelse{\equal{#1}{}}{\ensuremath{\mathit{EGAL}}\xspace}{\ensuremath{\mathit{EGAL}(#1)}\xspace}}
\renewcommand{\epsilon}{\varepsilon}
\newcommand{\profiles}{\mathcal{C}}

\title{\vspace{-3ex}Funding Public Projects:\\ A Case for the Nash Product Rule} 

\author{Florian Brandl\\ Bonn  \and Felix Brandt\\ TUM \and Matthias Greger\\ TUM \and Dominik Peters\\ Toronto  \and Christian Stricker\\ TUM \and Warut Suksompong\\ NUS}

\begin{document}

\maketitle

\begin{abstract}
We study a mechanism design problem where a community of agents wishes to fund public projects via voluntary monetary contributions by the community members. This serves as a model for public expenditure without an exogenously available budget, such as participatory budgeting or voluntary tax programs, as well as donor coordination when interpreting charities as public projects and donations as contributions. Our aim is to identify a mutually beneficial distribution of the individual contributions. In the preference aggregation problem that we study, agents report linear utility functions over projects together with the amount of their contributions, and the mechanism determines a socially optimal distribution of the money. We identify a specific mechanism---the Nash product rule---which picks the distribution that maximizes the product of the agents' utilities. This rule is Pareto efficient, and we prove that it satisfies attractive incentive properties: it spends each agent's contribution only on projects the agent finds acceptable, and agents are strongly incentivized to participate. 

\end{abstract}


\section{Introduction}
\label{sec:intro}

Italian tax payers have the option, under the \emph{cinque per mille} program, to redirect 0.5\% of their personal income tax to a non-profit organization of their choice. To participate, tax payers enter an organization's tax code into their tax return, choosing from a catalog of about 600 research organizations, 10,000 sports organizations, or 47,000 voluntary organizations.\footnote{\url{https://www.agenziaentrate.gov.it/portale/web/guest/elenchi-versione-precedenti}}
Participating in this program is a good choice for anyone who believes that funding for at least one of these organizations would do more good than additional tax income to the Italian government. In 2017, more than 10 million tax payers participated, for a total payout of more than 300 million euros.\footnote{\url{http://www.vita.it/it/article/2019/03/27/5-per-mille-2017-ecco-gli-elenchi-a-un-passo-dai-500-mln/151058/}}
However, one might worry that the allocation of funding to the organizations is inefficient because too little information about the participants' preferences is elicited.
Each person only indicates a single organization, but presumably they would be happy to support any of several organizations.
If we knew this approval information, we would likely be able to find an allocation that everyone prefers, in the sense that the money directed to approved organizations would be larger for each tax payer.

Suppose we convinced the Italian government to allow tax returns to indicate a list of organizations rather than just one. Given this information, how should we decide on the allocation of funds? A simple way to ensure a Pareto efficient outcome would be to maximize utilitarian welfare: one could define an individual's welfare as the amount of money disbursed to approved organizations and then maximize the sum of the welfare of each participating tax payer.  
The result would be that all the available funds would be disbursed to the (usually unique) organization that received the most ``votes''. While this is efficient, it fails to provide the participation incentives of the current system: one additional vote is unlikely to change which organization is most popular, and those who do not think that this organization is worth funding will choose to not participate.

It is in fact quite difficult to find an allocation mechanism that retains the strong participation incentives of the naive system (where each agent chooses just one organization) and also selects an efficient outcome.
Mechanisms that spend each voter's contribution only on approved organizations tend to fail efficiency. To narrow down the search, we can observe that any mechanism that incentivizes participation must also satisfy some natural group fairness axioms: for example, if a group of voters all approve the same list of organizations, then the mechanism must spend the accumulated tax contribution of the group on organizations on this list. A result by \citet{BMS02a} about group fairness implies that among separable social welfare functions, there is only a single candidate that might work: maximizing the \emph{Nash product}, which selects the allocation of funds that maximizes the product (rather than the sum) of utilities. In this paper, we prove that the Nash product rule indeed incentivizes participation at least as much as the naive rule. In particular, if a tax payer chooses to participate by submitting a list of approved organizations, we can guarantee that the money allocated to those organizations grows by at least her individual tax contribution. In fact, it can grow by more than that, because the Nash product rule may choose to redirect others' contributions to these organizations.

This result makes the Nash product rule an attractive choice in many other contexts, where we wish to incentivize voluntary monetary contributions to a common pool which is to be spent on public projects in an efficient manner.
Examples of communities that face this problem might include residents of an apartment complex (who want to coordinate spending on gardening in a courtyard, or on cleaning services), homeowners on a city street (to coordinate tree care, snow removal, or security patrols), or student clubs in a university (to coordinate funding for events and meet-ups).

Charitable donations provide another important application. Typically, these are undertaken independently without coordination among donors. As a consequence, mutual interest in the same charities goes unnoticed, even though the utility of all donors could be increased through coordination. For any given community of donors, such as the employees of a company running an annual charity matching program%
\footnote{For example, Microsoft and Apple run such programs. In 2020, Microsoft employees donated over \$110 million to charities, which the company doubled to over \$220 million.
 (\url{https://www.microsoft.com/en-us/corporate-responsibility/philanthropies/employee-engagement}).
Since its inception in 2011, Apple's charity matching program raised nearly \$600 million in total donations for more than 34,000 organizations (\url{https://www.apple.com/newsroom/2020/12/a-landmark-year-of-giving-from-apple/}).
} 
or donors using charitable giving mechanisms like so-called ``donor-advised funds'' offered by the same asset manager%
\footnote{For example, the asset manager Fidelity Charitable made over \$9 billion in donor-recommended grants in 2020 to 170,000 organization via donor-advised funds (\url{https://www.fidelitycharitable.org/insights/2021-giving-report.html}).
}, 
introducing a voting system based on the Nash product rule could produce better outcomes.
The same conclusion holds in higher-stakes applications involving major philanthropic foundations. Notably, the Open Philanthropy Project, which grants more than \$100 million a year to various organizations, has called on academics to develop mechanisms to combine different staff members’ views on the most effective giving opportunities, and also to help coordinate the giving of different philanthropic organizations \citep{Mueh17a}. Interest in donor coordination mechanisms has also been expressed in the effective altruism community \citep{Pete19b}, many of whose members have pledged to donate 10\% of their income to effective charities. \\

While our full model allows agents to specify fine-grained utilities, for ease of exposition, we will start by assuming that each agent only submits a list of approved projects and that the agent is indifferent among these.
We also ask each agent $i$ to commit a monetary amount $C_i \in [0,B_i]$ to contribute to the funding system within her own personal budget $B_i$. 
The approvals are interpreted as dichotomous utility functions, so that $u_i(x) \in \{0,1\}$ is the utility per unit of money that agent $i$ assigns to project $x$. For a distribution $\delta$ of the overall collected contributions to projects, agent $i$’s utility of distribution $\delta$ is defined as $u_i(\delta) := \sum_{x} u_i(x)\delta(x)$, where $\delta(x)$ is the amount spent on project $x$. With dichotomous utilities, $u_i(\delta)$ is just the total amount of money that $\delta$ spends on projects approved by $i$.

Our main result concerns contribution incentives. We wish to assure agents that it is beneficial for them to contribute their whole budget to the mechanism. 
In other words, if $i$ contributes an additional amount $\epsilon > 0$ of money, then the total amount spent on projects approved by $i$ needs to increase by at least $\epsilon$.
Formally, if the mechanism selects distribution $\delta$ when $i$ contributes $C_i \in [0,B_i-\epsilon]$, and selects distribution $\delta'$ when $i$ contributes $C_i + \epsilon$ with $\epsilon > 0$, then we must guarantee that $u_i(\delta') \ge u_i(\delta) + \epsilon$.
We call this property \emph{contribution incentive-compatibility}.

While this property may seem mild on first sight, it is difficult to satisfy together with efficiency.
A naive procedure where each agent's contribution is split uniformly between her approved projects, violates efficiency. Maximizing utilitarian welfare among all distributions is efficient, but severely violates contribution incentive-compatibility, because the mechanism may spend agent $i$'s contribution on projects that are not acceptable to $i$. If we constrain the welfare maximization to distributions where each agent's contribution is only spent on acceptable projects (aka the ``conditional utilitarian rule''), contribution incentive-compatibility is satisfied but we lose efficiency. Replacing the utilitarian objective with a Rawlsian leximin objective does not work either, and mechanisms based on serial dictatorships fail, too.

Remarkably, the Nash product rule described above combines efficiency and contribution incentives. The Nash product rule selects the distribution $\delta$ that maximizes $\prod_{i\in N} u_i(\delta)^{C_i}$, where $C_i$ is the size of $i$’s contribution and $N$ is the set of agents. Since this rule maximizes a monotonic function of agents' utilities, its outcome is guaranteed to be efficient. 
While it is easy to see that the Nash product rule satisfies $u_i(\delta') \ge u_i(\delta)$ when $i$ contributes an additional amount $\epsilon$, it is more difficult to establish that we have $u_i(\delta') \ge u_i(\delta) + \epsilon$ as required by contribution incentive-compatibility.
Our main result, \Cref{thm:nashir}, shows that this property is satisfied by the Nash product rule.
The proof reasons about the trajectory of the maximizer of the Nash product as a function of agent $i$’s contribution. 
We derive a lower bound for the derivative of agent $i$'s utility along this trajectory. 
Integrating this bound yields the result.

The Nash product rule is the only mechanism known to us that is both efficient and contribution incentive-compatible for an arbitrary number of projects. It is plausible that the mechanism is characterized by these properties, but we could not establish this. 
However, as discussed in \Cref{sec:related}, two characterizations by \citet{BMS02a} and \citet{GuNe14a} imply that, at least when imposing further strong assumptions, the Nash product rule is characterized by contribution incentive-compatibility. 
The Nash product rule satisfies an additional property that is important in our model: it spends the contribution of user $i$ only on projects that are approved by $i$. Formally, we say that a distribution $\delta$ is \emph{decomposable} if we can decompose it as $\delta = \delta_1 + \cdots + \delta_n$ such that $\delta_i$ spends exactly $C_i$, and spends it only on projects acceptable to $i$. In \Cref{thm:nashimp}, we prove that the distribution selected by the Nash product is always decomposable.
This provides an intuitive justification of the chosen distribution to agents. 
Decomposability can also be interpreted as an incentive property: in some applications, the money is not distributed to projects by a central clearinghouse, but the mechanism's output is just used as a recommendation to agents where to direct their contribution.
In this case, decomposability becomes essential, since a recommendation to send money to an unacceptable project is likely to be ignored. There do exist artificial mechanisms other than the Nash product rule that are efficient and decomposable,%
\footnote{To construct such mechanisms, note that we can modify an efficient distribution and retain efficiency as long as we do not increase its support \citep[see, e.g.,][]{ABB14b}.
So we can, for example, take the support of the Nash product distribution and let every agent assign her entire contribution to one of her most preferred projects within the support. The resulting distribution is efficient and decomposable.} 
but the Nash product rule is the only such mechanism known to us that arises naturally from the maximization of a social welfare function.
The Nash product rule not only satisfies decomposability, but, moreover, the fraction of agent $i$'s contribution to project $x$ is directly proportional to the utility $\delta(x)u_i(x)$ she derives from $x$ in the Nash product distribution $\delta$.
Thus, agents can then easily compute their individual distributions $\delta_i$ once they know $\delta$. We leverage this observation to construct a simple, dynamic procedure in which agents iteratively revise their contributions to projects in proportion to the utility they receive from the distribution of the previous round.
A result of \citet{Cove84a} from the theory of optimal portfolio selection implies that this procedure approximates a Nash product distribution arbitrarily well as the number of rounds goes to infinity. In \Cref{thm:dynamicNash}, we state this result and give a compact proof tailored to our setting. Hence, the Nash product rule arises naturally from a simple decentralized spending dynamic.

Our results generalize beyond the case of dichotomous utilities. 
In our formal treatment, we allow agents to indicate arbitrary utility values $u_i(x) \ge 0$ for the projects, and extend these to distributions as linear utilities as before, so $u_i(\delta) = \sum_{x} u_i(x)\delta(x)$. 
The Nash product rule works for this more general class of utilities, and in particular it retains efficiency.
It also continues to satisfy contribution incentive-compatibility and decomposability, but these two properties hold in a weak sense that only distinguishes \emph{acceptable} projects with strictly positive utility $u_i(x) > 0$ from \emph{unacceptable} ones with $u_i(x) = 0$. Contribution incentive-compatibility guarantees that the amount spent on acceptable project grows by $\epsilon$ when an extra amount of $\epsilon$ is contributed, and decomposability guarantees that an agent's contribution is only spent on acceptable projects.
In \Cref{sec:limits}, we discuss strengthened versions of decomposability and contribution incentive-compatibility that use fine-grained utilities and provide guarantees based on an agent's most-preferred project (for example, strong decomposability requires that an agent's contribution is only spent on most-preferred projects). These two stronger properties may be desirable, but we prove impossibility theorems that show that each of the two strengthened axioms is incompatible with efficiency.

While the Nash product rule is incentive-compatible in the sense that it is decomposable and incentivizes contribution, it still allows for other strategic behavior. In particular, agents may have an incentive to misrepresent their utility functions. Because the Nash product penalizes distributions in which some agents obtain very low utility, it can be beneficial for agents to pretend to like popular projects less, or even to mark them as unacceptable. This will make the Nash product rule worry that those agents will be underserved, and thus increase the funding of other projects acceptable to them. Unfortunately, by a result due to \citet[Thm.~2]{Hyll80a}, \emph{every} efficient mechanism will be vulnerable to misrepresentation of preferences, except for dictatorships. This impossibility is robust, and analogues hold even for dichotomous utilities \citep{BMS05a,Dudd15a,BBPS21a}. Since efficiency is our main objective, we ignore possible misrepresentation of preferences in our discussion.

Overall, our discussion suggests that the Nash product rule is a prime candidate for funding public projects through voluntary individual contributions. It combines efficiency with strong incentive properties, and as detailed in \Cref{sec:related}, it also satisfies important fairness and proportionality properties.
Finally, the rule is simple to define, can be easily approximated, and because it is decomposable, its distribution decisions can be easily understood by users. We are excited for the possibility of implementing a system based on the Nash product rule in the real world.

\section{Related Work}
\label{sec:related}

The classic literature on \emph{private provision of public goods} \citep[e.g.,][]{Samu54a,BBV86a} studies Nash equilibria in the non-cooperative setting where each agent decides how much to contribute to funding a public good. 
The main conclusion is that public goods will be underprovided in equilibrium, leading to inefficiency. 
In our model, we study cases where underprovision is less of a problem, for example because a company's matching program makes contributing a dominant strategy, or because the outside option is unattractive, such as paying more taxes. Similarly, in the context of donor coordination, agents may have set aside a part of their income as a budget for charitable activities.
The inefficiency that we are worried about is an inefficient allocation among different public goods. 

In contrast to the above literature, we study a setting where there is an explicit coordinating infrastructure or mechanism that aggregates preferences. Our model can thus be said to fall within the area of collective decision making where the set of alternatives is some subset of the Euclidian space, modeling divisible public goods or lotteries over indivisible public goods \citep[see, e.g.,][]{LeWe11a}.
Two concrete applications in this context that have recently gained a lot of attention are those of \emph{participatory budgeting} \citep[e.g.,][]{AzSh20a} and \emph{probabilistic social choice} \citep[e.g.,][]{Bran17a}. 

Participatory budgeting is a paradigm that allows citizens to collectively decide how a portion of a public budget ought to be spent \citep{Caba04a}. It has mostly been studied under the assumption that the budget is provided by an outside source (such as the city government).
In the most common model, projects come with a fixed cost, and they can either be fully funded or not at all. 
Probabilistic social choice studies the aggregation of individual preferences into a lottery over alternatives. Both settings are interrelated because a division of a fixed endowment among projects is equivalent to a probability distribution over alternatives. The social choice literature typically focusses on ordinal preferences. \citet{BMS05a} have initiated the study of probabilistic social choice for dichotomous utility functions, where cardinal and ordinal preferences coincide. 

The idea of maximizing the product of agents' utilities originates in the \emph{Nash bargaining solution} and the corresponding mechanism is therefore often referred to as the Nash product rule \citep{Nash50b}.\footnote{In the context of asset allocation, this rule is known as the Kelly criterion \citep{Kell56a}. When interpreting the utility vectors of the agents as a multidimensional random variable which takes the value of agent $i$'s utility vector with probability $C_i$, the Kelly criterion maximizes the same objective function as the Nash product rule.} The Nash product rule has recently become popular in various fields, including the allocation of indivisible private items \citep{CKM+19a}, committee elections \citep{LaSk18b}, and participatory budgeting \citep{FGM16a,FMS18a}. For all these settings, the Nash product rule satisfies strong fairness and proportionality properties.

In the context of dichotomous preferences, \citet{ABM17a} showed that the Nash product rule guarantees \emph{average fair share}: for any group of $\alpha\%$ of the agents which is \emph{cohesive} (there is a project that they all approve), for an average group member, the Nash product spends at least $\alpha\%$ of the endowment on approved projects. They also proved that the Nash product rule satisfies \emph{strict participation}. This property, which was introduced by \citet{BBH15b}, makes sense for the fixed-endowment setting, but it is relatively weak in our setting with variable contributions.\footnote{For dichotomous preferences, strict participation implies that if before contributing, $\beta\%$ of others' money was spent on $i$'s approved projects, then strictly more than $\beta\%$ of $i$'s additional contribution will be spent on $i$'s approved projects, while others' money is not spent in a worse way for $i$.
On the other hand, contribution incentive-compatibility ensures that if agent $i$ contributes money, \emph{all} of it will be spent on $i$'s approved projects, while again others' money is not spent in a worse way for $i$. 
\citeauthor[][p.~768,]{ABM17a} mention that a large class of additive welfarist rules satisfy strict participation. Out of these only the Nash product rule satisfies contribution incentive-compatibility.} 
Our main result showing that the Nash product rule is contribution incentive-compatible implies \citeauthor{ABM17a}'s \citeyearpar{ABM17a} result.

Two axiomatic characterizations of the Nash product rule are of particular interest in our context. 
First, \citet[Prop.~6]{BMS02a} have shown that the Nash product rule is the only rule that satisfies \emph{unanimous fair share} (a condition weaker than decomposability and significantly weaker than contribution incentive-compatibility) among rules that maximize a quantity of the form $\sum_{i\in N} C_i f(u_i(\delta))$ for some function $f$ \citep[see also][p.~768]{ABM17a}.
Their result was shown for the domain of dichotomous preferences, but easily extends to our more general domain due to the restricted form of mechanisms considered. 
Secondly, \citet{GuNe14a} have characterized a solution concept called the \emph{diversity value} for weighting different information sources based on their reliability.
Their result can be translated into a characterization of the Nash product rule for dichotomous preferences using conditions such as convexity, continuity, reinforcement, and a core condition that is again weaker than decomposability and significantly weaker than contribution incentive-compatibility. 

\citet{FGM16a} have initiated the study of a participatory budgeting setting where projects can receive an arbitrary amount of funding (like in our paper) but the budget is still exogenous and of fixed size. \citeauthor{FGM16a} argued that allocations in \emph{Lindahl equilibrium} \citep{Fole70a} are particularly desirable. The Lindahl equilibrium is a market equilibrium in an artificial market for public goods. In these markets, each agent faces personalized prices (usually interpreted as taxes) for the public goods, and in equilibrium each agent demands the same bundle of public goods. Under standard assumptions, \citet{Fole70a} showed that a Lindahl equilibrium exists (by reducing to the Arrow--Debreu private goods case), and is efficient. He also showed that equilibrium allocations are in the \emph{core}: no coalition of agents can afford (using only a fraction of the budget proportional to their size) an allocation that each coalition member prefers to the equilibrium. For the case of additive linear utilities, \citet{FGM16a} proved that the Nash product rule yields an allocation in Lindahl equilibrium, and hence is in the core.\footnote{This mirrors the canonical result that the Nash product yields an equilibrium in Fisher markets for private goods under additive valuations \citep{EiGa59a}.} The core can be interpreted as guaranteeing agents proportional representation: if a fraction of $\alpha\%$ of agents assign positive utility only to some set $A'$ of projects, then the Nash product rule will spend at least $\alpha\%$ of the budget on projects in $A'$. 

\citet{GuPe20a} study Lindahl equilibrium as a collective choice rule.
They characterize the set of all Lindahl equilibrium utility profiles as the outcomes of a bargaining solution they call the equitable solution.
Every outcome of the equitable solution can be justified by being the Nash bargaining outcome of a simple related bargaining problem.

The key difference between all of the above literature and our model is that in our model, the individual contributions to the pool are \emph{owned} by the agents. This suggests the definitions of the axioms of decomposability and  contribution incentive-compatibility, which---to the best of our knowledge---have not been considered in previous work.

Since we study a model of public goods provision with direct monetary contributions, one could assume quasilinear utilities and try to use the Vickrey-Clarke-Groves (VCG) mechanism. However, since the VCG mechanism implements the utilitarian rule, it will not incentivize contributions in our sense. It will also not be budget-balanced and generally run a deficit. Thus, the VCG mechanism does not seem useful for our purposes.

\section{Model and Axioms}
\label{sec:prelims}

Let $A$ be a finite set of $m$ \emph{projects} (e.g., charities or joint activities).
A \emph{distribution} $\delta$ is a function that describes how some amount $V$ 
is distributed among the projects, so $\delta: A \rightarrow \mathbb{R}_{\ge 0}$ with $\sum_{x\in A} \delta(x) = V$. 
For convenience, we write distributions as linear combinations of projects, so that $a + 2\,b$ denotes the distribution $\delta$ with $\delta(a)=1$ and $\delta(b)=2$.
The set of all distributions of value $V$ is denoted by $\Delta(V)$.

There is a finite set $N$ of $n$ \emph{agents}. 
Each agent $i\in N$ has a \emph{budget} $B_i\in \mathbb R_{>0}$ and a \emph{utility function} $u_i\colon A \rightarrow \mathbb{R}_{\geq 0}$, where $u_i(a) \ge 0$ is agent $i$'s utility for every unit of money that goes to project $a$.
So agent $i$'s utility for a distribution $\delta\in\Delta(V)$ is 
\[u_i(\delta)=\sum_{x\in A} \delta(x) \cdot u_i(x)\text.\]
A project is said to be \emph{acceptable} by an agent if it gives her positive utility, and \emph{unacceptable} it gives her utility $0$.
In the special case that an agent assigns the same utility to all projects, we label all projects as acceptable and set all utilities to $1$.
For convenience, we rescale utility functions such that the utility assigned to least-preferred acceptable projects is $1$, i.e., $\min\{u_i(x)\colon u_i(x) > 0\} = 1$.
(We explain in \Cref{fn:normalization}, which follows \Cref{def:contribution-ic}, how to adapt the model to work without this normalization.)
A utility function $u_i$ is \emph{dichotomous} if $u_i(x)\in\{0,1\}$ for all $x\in A$, so that agent $i$ only distinguishes between acceptable and unacceptable projects without discriminating between the acceptable ones. In this context, we refer to the set of acceptable projects of an agent as her \emph{approval set}.

Each agent chooses a non-negative \emph{contribution} $C_i\in[0,B_i]$ no larger than her budget that she contributes to a common pool.
A \emph{(contribution) profile} is a tuple of contributions $\prof = (C_i)_{i\in N}$.
Let $\profiles = \Pi_{i\in N}[0,B_i]$ denote the set of all profiles, and let $\profiles_{>0} = \Pi_{i\in N}(0,B_i]$ be the set of profiles where every agent has a positive contribution.
The sum of all agents' contributions in a profile is $\csum = \sum_{i\in N} C_i$ and is called the \emph{pool}.
A \emph{mechanism} $f$ maps a profile $\prof$ to a distribution of the pool $f(\prof) \in \Delta(\csum)$. 
Hence, we take the agents' budgets and utility functions to be fixed and known and consider the game induced by a mechanism that asks the agents for their contributions.  

We now discuss the main properties of distribution mechanisms that we are interested in: efficiency, decomposability, and  contribution incentive-compatibility.

A mechanism that yields high-quality distributions should, at minimum, satisfy Pareto efficiency. 
Indeed, if a mechanism produces a distribution so that we could redistribute the pool between projects and thereby increase the utility of every agent who contributes to the mechanism, then the mechanism has not made full use of the potential for mutual gains.
We intend mechanisms to ignore agents with zero contributions, and therefore define efficiency only with respect to agents with positive contributions. Thus, a Pareto improvement may be worse for an agent who has chosen not to contribute to the mechanism.

\begin{definition}[Efficiency]
Given a contribution profile $\prof \in \profiles$, a distribution $\delta'\in \Delta(\csum)$ \emph{dominates} another distribution $\delta\in\Delta(\csum)$ if $u_i(\delta')\ge u_i(\delta)$ for all $i \in N$ with $C_i > 0$ and $u_i(\delta')> u_i(\delta)$ for some $i \in N$ with $C_i > 0$.
A mechanism $f$ is \emph{efficient} if for every profile $\prof$, no distribution dominates $f(\prof)$.
\end{definition}

In applications, the mechanism might operate in a decentralized setting and not be able to directly control the use of the agents' contributions (for example, when a donor coordination service does not actually collect money from its participants). 
In such cases, the mechanism's output $\delta$ is better understood as a recommendation to the agents about how they should use their resources.
We would then need to decompose $\delta$ into individual distributions $\delta_i\in\Delta(C_i)$, so that if every agent spends her reported contribution according to $\delta_i$, we recover $\delta$.
A distribution is decomposable if $\delta_i$ spends agent $i$'s contribution exclusively on projects acceptable by $i$.

\begin{definition}[Decomposability]
\label{def:decomposable}
Let $\prof$ be a profile.
A distribution $\delta\in \Delta(\csum)$ is \emph{decomposable} if it can be divided into individual distributions $(\delta_i)_{i \in N}$ with $\delta_i \in \Delta(C_i)$ for all $i\in N$ and $\delta = \sum_{i\in N} \delta_i$ such that for all $i\in N$, we have $\delta_i(x) > 0$ only if $u_i(x) > 0$.
\end{definition}

We say that a mechanism $f$ is \emph{decomposable} if $f(\prof)$ is decomposable for all profiles $\prof$. 

In \Cref{sec:limits}, we discuss a strengthening of decomposability which requires that $\delta_i(x) > 0$ only if $i$ has assigned maximum utility to $x$, i.e., only if $u_i(x) \ge u_i(y)$ for all $y \in A$. However, this requirement turns out to be too strong; it clashes with efficiency.
Alternative characterizations of decomposability and strong decomposability are given in \Cref{app:decomp}.

We want to incentivize agents to contribute their entire budget since this increases the potential gains from coordination.
Suppose each agent $i$ aims to maximize $u_i(f(\prof)) - C_i$, i.e., her utility for the distribution of the pool minus her own contribution.
This objective is well-motivated if agent $i$ could spend money outside the mechanism so as to obtain one unit of utility per unit of money. Given our normalization of utility functions, this is equivalent to agent $i$ valuing one unit of money as much as one unit of money going to a least-preferred acceptable project.
A mechanism then incentivizes agent $i$ to contribute her entire budget if choosing $C_i = B_i$ is a weakly dominant strategy for agent $i$.
If this property holds independently of the agents' budgets, it is equivalent to $u_i(f(\prof)) - C_i$ being weakly increasing in $C_i$.
We call such a mechanism contribution incentive-compatible. 

\begin{definition}[Contribution incentive-compatibility]
\label{def:contribution-ic}
A mechanism $f$ is \emph{contribution incentive-compatible} if for each $i\in N$ and all profiles $C$, we have
\[
u_i(f(C_{-i}, C'_i)) - C'_i \le u_i(f(C_{-i}, C_i)) - C_i \quad\text{for all $C'_i$ with $0 \le C_i' \le C_i$.}\footnote{\label{fn:normalization}We normalized utility functions so that $\min\{u_i(x)\colon u_i(x) > 0\} = 1$. Without this normalization, the definition of contribution incentive-compatibility would read $u_i(f(C_{-i}, C'_i)) - C'_i\min\{u_i(x)\colon u_i(x) > 0\} \le u_i(f(C_{-i}, C_i)) - C_i\min\{u_i(x)\colon u_i(x) > 0\}$.}
\]
\end{definition}
In particular, not participating ($C_i' = 0$) is at least weakly dominated by contributing any positive amount of one's own budget.
We can re-write the definition as
\[
u_i(f(C_{-i}, C_i)) - u_i(f(C_{-i}, C_i - \epsilon)) \ge \epsilon
\]
for all $0\le\epsilon\le C_i$.
Thus, increasing one's contribution by $\epsilon$ causes an increase of at least $\epsilon$ in the utility derived from the distribution selected by $f$.

In \Cref{sec:limits}, we discuss a strengthening of contribution incentive-compatibility that requires $u_i(f(\prof)) - C_i\cdot\max_{y\in A}u_i(y)$ to be weakly increasing in $C_i$.
This corresponds to the assumption that an agent values one unit of money as much as one unit of money going to her highest utility project.
Again, this stronger version is incompatible with efficiency.

Decomposability and contribution incentive-compatibility are logically independent properties, even when utilities are dichotomous. 
\Cref{app:contributionnotdecomposable} gives two mechanisms that satisfy only one of these axioms at a time. Nevertheless, the two properties seem to be related as together with efficiency, contribution incentive-compatibility is likely to imply decomposability since the Nash product rule always returns a decomposable distribution.

\section{The Nash Product Rule}\label{sec:nash}

The \emph{Nash product}, which refers to the product of agent utilities, is often seen as a compromise between utilitarian and egalitarian welfare \citep[][]{Moul88a}. Maximizing the Nash product has been found to yield fair and proportional outcomes in many preference aggregation settings, and it also turns out to be attractive in our context. Formally, 
\[
\nash(\prof) = \argmax_{\delta \in \Delta(\csum)} \prod_{i\in N} \left( u_i(\delta)\right)^{C_i} = \argmax_{\delta \in \Delta(\csum)} \sum_{i\in N} C_i \log\left( u_i(\delta) \right)\text.\footnote{\nash is invariant to rescaling utility functions. Hence, the normalization $\min\{u_i(x)\colon u_i(x) > 0\} = 1$ does not affect \nash.}
\]
Note that $\nash$ weights agents by their contribution. (As a convention, we let $0^0 = 1$ and $0 \log 0 = 0$, so that $\nash$ ignores agents with zero contribution.) An unweighted Nash rule where each agent gets assigned the same weight would violate decomposability and contribution incentive-compatibility. Indeed, that mechanism would not take into account the individual contributions at all, and thus agents with large contributions would have the same influence as agents with very small (or even zero) contributions. This shows the need to weight agents.

There can be several distributions that maximize the Nash product.\footnote{Consider the following example (which notably does not contain any `clone' projects). There are four agents with approval sets $\{a,c\}$, $\{a,d\}$, $\{b,c\}$, and $\{b,d\}$ and each agent contributes 1. Then, the set of \nash distributions consists of all convex combinations of $2a + 2b$ and $2c + 2d$.} 
However, all of these distributions yield the same amount of utility to each agent (due to the strict convexity of the objective function, see \Cref{lem:nashmax}). Thus, we can arbitrarily break ties in these cases without affecting any of the axioms considered here.

We now show that $\nash$ is efficient, decomposable, and incentivizes contribution.
The first of these is easy: The distribution $\nash(\prof)$ maximizes a sum of functions, namely $C_i\log(\cdot)$, that are strictly increasing in the agents' utilities provided that $C_i > 0$. Thus, $\nash$ is efficient \citep[see, e.g.,][]{Moul88a}. We will prove that $\nash$ satisfies the other two axioms later in this section. First, we verify these claims for a small example.

\begin{example}
\begin{table}[htb]
	\centering

	\[
	\begin{array}{rcc@{\hskip 3em}c}
		\toprule
		& u_i(a) & u_i(b) & C_i\\\midrule
		\text{Agent 1} & 1 & 0 & 1\\
		\text{Agent 2} & 1 & 3 & 1\\
		\bottomrule
	\end{array}
	\]

	\caption[Type profile]{Profile $\prof = (1,1)$ with $B_i=C_i$ for $i \in \{1,2\}$ and $\nash(\prof) = 1.5\ a + 0.5\ b$.}
	\label{tab:nashexample}
\end{table}

A simple example of \nash for a profile $\prof$ with two agents and two projects is shown in \Cref{tab:nashexample}.
We have \[\delta = \nash(\prof) = \argmax_{\delta \in \Delta(2)} \delta(a) \cdot (\delta(a) + 3\ \delta(b)) = 1.5\ a + 0.5\ b\text.\]
The collective distribution $\delta$ can be decomposed into individual distributions 
\[
\delta_1=a \quad\text{and}\quad \delta_2=0.5\ (a + b)\text.
\] 
Contribution incentive-compatibility is satisfied in this example because 
\begin{eqnarray*}
u_1(\nash((1-\epsilon_1,1)))+\epsilon_1 &=& 1.5-0.5\epsilon_1 \text{ and}\\ 
u_2(\nash((1,1-\epsilon_2)))+\epsilon_2 &=& 6-2\epsilon_2-2\min\left\{1.5,2-\epsilon_2 \right\}\text.
\end{eqnarray*}
are (weakly) decreasing for increasing $\epsilon_1=B_1-C_1$ and $\epsilon_2=B_2-C_2$, respectively.
On the other hand, simply maximizing the sum of individual utilities in this example would result in $\delta' = 2\ b$, which is not decomposable, as project $b$ is unacceptable for agent 1, and violates contribution incentive-compatibility because agent 1 would prefer an outside option to participating in the mechanism. \qed
\end{example}

\subsection{Decomposability}
The Nash product distribution is the solution of an optimization problem, and thus satisfies the first-order conditions of optimality.
By manipulating these conditions, we can show that the Nash product distribution is always decomposable.\footnote{This proof is similar to a result by \citet{GuNe14a} who consider \nash with dichotomous preferences, and establish an equivalent property in this restricted setting.}

\begin{theorem}\label{thm:nashimp}
\nash is decomposable.
\end{theorem}
\begin{proof}
	We have to show that there is a decomposition of $\nash(\prof)$ into $\delta_i\in\Delta(C_i)$, $i\in N$, such that $\sum_{i\in N}\delta_i(x) = \delta(x)$ for all $x$.
	
We consider the Karush--Kuhn--Tucker (KKT) conditions
and write the Lagrangian as
\[
\mathcal L(\delta, \lambda, \mu_1, \dots, \mu_m) = \sum_{i\in N} C_i \log\left( u_i(\delta) \right) + \lambda \left(\csum - \sum_{x\in A} \delta(x)\right) + \sum_{x \in A} \mu_x \delta(x) \text,
\]
where $\lambda \in \mathbb R$ is the Lagrange multiplier for the constraint $\sum_{x\in A} \delta(x) = \csum$ and $\mu_x \ge 0$ is the multiplier for the constraint $\delta(x) \ge 0$.

Suppose $\delta$ is an optimal solution. By complementary slackness, we must have $\mu_x = 0$ whenever $\delta(x) > 0$. Also, we must have $\partial \mathcal L / \partial \delta(x) = 0$, that is, $\sum_{i\in N} C_iu_i(x)/u_i(\delta) - \lambda + \mu_x = 0$. 
By case distinction based on whether $\delta(x) > 0$, it follows that $\lambda\delta(x) = \sum_{i\in N} C_i\delta(x)u_i(x)/u_i(\delta)$ for all $x\in A$. Hence,
\[ \lambda \cdot \csum = \sum_{x\in A} \lambda\delta(x) = \sum_{x\in A} \sum_{i\in N} C_i \frac{\delta(x)u_i(x)}{u_i(\delta)} =  \sum_{i\in N} C_i \frac{u_i(\delta)}{u_i(\delta)} = \sum_{i\in N} C_i = \csum\text.  \]
So $\lambda = 1$, and hence $\sum_{i\in N} C_iu_i(x)/u_i(\delta) = 1$ for all $x\in A$ such that $\delta(x) > 0$.

Now, for each $i\in N$, define an individual distribution $\delta_i \in \Delta(C_i)$ with $\delta_i(x) = C_i\delta(x)u_i(x)/u_i(\delta)$ for all $x\in A$. 
Clearly, $\supp(\delta_i) \subseteq\{a\in A\colon u_i(a) > 0\}$ and $\delta_i\in\Delta(C_i)$, since $\sum_{x\in A} \delta(x)u_i(x) = u_i(\delta)$. 
To see that $\delta = \sum_{i \in N} \delta_i$, note that for $x\in A$ with $\delta(x) = 0$ we have $\delta_i(x) = 0$ for all $i\in N$, and for $x\in A$ with $\delta(x) > 0$, we have
\[
\sum_{i \in N} \delta_i(x) = \sum_{i \in N} C_i\delta(x) \frac{u_i(x)}{u_i(\delta)} = \delta(x) \sum_{i \in N} C_i \frac{u_i(x)}{u_i(\delta)} = \delta(x)\text. \qedhere
\]

\end{proof}

By inspecting the proof, we see that the distribution $\delta_i$ of agent $i$ satisfies a stronger notion of decomposability: the fraction of her contribution that she gives to project $x$ is proportional to the utility $\delta(x)u_i(x)$ she derives from $x$ in the Nash product distribution $\delta$ \citep[see also][]{GuNe14a}.
For example, if half of agent $i$'s utility $u_i(\delta)$ is due to the amount $\delta(x)$ spent on $x$, then she transfers half of her contribution to $x$.
Thus, it suffices that a central clearinghouse announces the overall distribution $\delta$. Agents can then easily compute their individual distributions $\delta_i$ without needing to know the other agents' utility functions or contributions.

\subsection{Computation}
\label{sec:computation}

In general, \nash can be computed to arbitrary precision using convex programming \citep[see, e.g.,][]{BMS05a}.
However, \nash cannot be computed exactly (in the standard binary representation) because it may return distributions with irrational values. An example is given in \Cref{tab:nashexampleirr}. 
\begin{table}[htb]
	\centering

	\[
	\begin{array}{rccc@{\hskip 3em}c}
		\toprule
		& u_i(a) & u_i(b) & u_i(c) & C_i\\\midrule
		\text{Agent 1} & 1 & 1 & 0 & 1\\
		\text{Agent 2} & 1 & 0 & 1 & 1\\
		\text{Agent 3} & 0 & 1 & 1 & 1\\
		\text{Agent 4} & 0 & 0 & 1 & 1\\
		\bottomrule
	\end{array}
	\]

	\caption[Type profile]{Profile $\prof = (1,1,1,1)$ with approval sets $\{ab\},\{ac\},\{bc\},\{c\}$. Let $\delta = \nash(C)$. Alternatives $a$ and $b$ are symmetric, so $\delta(a) = \delta(b)$. Thus $\delta(c) = 4 - 2\delta(a)$. So we can write the Nash objective as $2\delta(a)(4-\delta(a))^2(4-2\delta(a))$, which is maximized for $\delta(a) = (7 - \sqrt{17})/4$.}
	\label{tab:nashexampleirr}
\end{table}

We observed after the proof of \Cref{thm:nashimp} that the distribution selected by \nash is a fixed point of a process where agents spend their contribution on a project in proportion to the utility they receive from that project under the \nash distribution.
This observation, due to \citet{GuNe14a}, gives rise to a simple, dynamic procedure for approximating \nash, similar to the \emph{proportional response dynamic} that converges to equilibrium in Fisher markets for private goods \citep{Zhan11a}. 

For $C\in\profiles_{>0}$, consider the mapping $f\colon\Delta(\csum)\rightarrow\Delta(\csum)$ defined by
\begin{equation*}
	(f(\delta))(x) = \sum_{i\in N} C_i \frac{u_i(x)}{u_i(\delta)} \delta(x) \qquad \text{for all $\delta \in \Delta(\csum)$.}\footnote{Note that $f$ is well-defined only if $u_i(\delta) > 0$ for all $i\in N$. This will always hold in our analysis.}
\end{equation*}
The $i$th summand is called the individual distribution of agent $i$.
Hence, given a distribution $\delta$, the fraction of the contribution agent $i$ assigns to project $x$ in $f(\delta)$ equals the fraction of the utility agent $i$ derives from the overall contribution $\delta(x)$ to $x$.
The proof of \Cref{thm:nashimp} shows that we have $f(\delta) = \delta$ for $\delta = \nash(\prof)$.
The mapping $f$ induces a dynamic procedure: For any initial distribution $\delta^{0}$, we obtain a sequence $(\delta^k)_{k\in\mathbb N}$ by setting $\delta^k = f(\delta^{k-1}) $ for each $k \ge 1$.

It turns out that this dynamic procedure has been studied in the literature on optimal portfolios, where projects correspond to stocks and utilities encode stock performance.%
\footnote{That literature has argued that a portfolio of stocks maximizing expected log returns (which corresponds to the Nash product) produces optimal earnings in the long run \citep[Chapter~16]{CoTh06a}. The formal analysis focusses on stock returns over time and thus does not seem relevant to the study of \nash as an aggregation rule.} 
In this context, \citet{Cove84a} showed that the Nash product of $\delta^k$ converges to the optimum Nash product if $\delta^0$ has full support, and the sequence $(\delta^k)_{k\in\mathbb N}$ converges to a Nash distribution under additional assumptions. Thus, by simply computing terms of the sequence $(\delta^k)_{k\in\mathbb N}$, one can approximate a Nash distribution without resorting to convex programming. 

For convenience, we give a compact proof of this result.  It is based on Cover's proof, which features a clever use of Jensen's inequality. Our proof is adapted to our setting and is more compact since our model assumes the number of agents to be finite. We emphasize that whenever the Nash distribution is unique (which it is for a generic profile), the sequence $(\delta^k)_{k\in\mathbb N}$ converges to it. 

We write $F(\delta) = \sum_{i\in N}\cnom_i \log(u_i(\delta))$ for the (log) Nash product of $\delta \in \Delta(\csum)$.

\begin{theorem}\label{thm:dynamicNash}
	Let $C\in\profiles_{>0}$ and $\delta^0\in\Delta(\csum)$ be a distribution with full support.
	Denote by $(\delta^k)_{k\in\mathbb N}$ its induced sequence.
	Then,	$(F(\delta^k))_{k\in\mathbb N}$ converges to the optimum Nash product. 
	If the Nash distribution is unique, $(\delta^k)_{k\in\mathbb N}$ 	    converges to \nash. 
\end{theorem}

\begin{proof}
	 Note that if $\delta$ has full support, then $u_i(\delta) > 0$ for all $i\in N$.
	 Moreover, in the next iterate $f(\delta)$, every agent assigns her contribution only to projects for which she has strictly positive utility.  
  Hence, $u_i(\delta^k)>0$ for all $i$ and $k$, and $\delta^k(x)=0$ for a project $x$ implies $u_i(x)=0$ for all agents $i$.
  We can thus ignore such projects and assume $\delta^k(x)>0$ for all $x$ and $k$.
	 Normalizing by dividing by $\csum$ if necessary, we may assume that $\csum = 1$, so that $\delta^k \in \Delta(1)$ for all $k$.
	 
    The proof proceeds in two steps.
    \begin{enumerate}
        \item The sequence $(F(\delta^k))_{k\in \mathbb{N}}$ converges.
        \item Every accumulation point of $(\delta^k)_{k\in \mathbb{N}}$ is a Nash product distribution. 
    \end{enumerate}
	\begin{step}\label{step:nashdyn1}
		For $k\geq 1$, we get
			\begin{align*}
			    F(\delta^{k+1}) - F(\delta^{k}) &= \sum_{i\in N} \cnom_i\log\left(\frac{u_i(\delta^{k+1})}{u_i(\delta^{k})}\right) =\sum_{i\in N}\cnom_i\log\left(\sum_{x\in A} \delta^{k+1}(x) \frac{u_i(x)}{u_i(\delta^{k})}\right) \\
			    &\overset{(1)}{=}\sum_{i \in N}\cnom_i\log\left(\sum_{x \in A}\left(\sum_{j \in N}\cnom_j\frac{u_j(x)}{u_j(\delta^k)}\right)\delta^k(x)\frac{u_i(x)}{u_i(\delta^k)}\right) \\
			    &\overset{(2)}{\geq} \sum_{i \in N}\cnom_i\sum_{x \in A}\delta^{k}(x)\frac{u_i(x)}{u_i(\delta^k)}\log\left(\sum_{j \in N}\cnom_j\frac{u_j(x)}{u_j(\delta^k)}\right) \\
			    &\overset{(3)}{=}\sum_{x \in A}\delta^k(x)\sum_{i \in N}\cnom_i\frac{u_i(x)}{u_i(\delta^k)}\log\left(\sum_{j \in N}\cnom_j\frac{u_j(x)}{u_j(\delta^k)}\frac{\delta^k(x)}{\delta^k(x)}\right) \\
			    &\overset{(4)}{=}\sum_{x \in A}\delta^{k+1}\log\left(\frac{\delta^{k+1}}{\delta^k}\right)\overset{(5)}{\geq}
			     \frac{1}{2\log(2)}\norm{\delta^{k+1}-\delta^k}_1^2 \geq 
			     0,
			\end{align*}
		    where $(1)$ and $(4)$ follow from the definition of the dynamic procedure,
		    $(2)$ is an application of Jensen's inequality for concave functions (notice that $\sum_{x \in A}\delta^{k}(x)\frac{u_i(x)}{u_i(\delta^k)}=1$),
		    $(3)$ changes the summation order,
		    and $(5)$ uses Lemma 11.6.1 of \citet{CoTh06a}, where the left-hand side is the Kullback-Leibler divergence of $\delta^{k+1}$ and $\delta^k$.   
		    
			 Hence, $(F(\delta^k))_{k\in \mathbb{N}}$ is a weakly increasing sequence.		  
			  As it is bounded from above by $F(\delta^*)$ where $\delta^*$ is a Nash product distribution, it converges.
	\end{step}
	\begin{step}
		The KKT-conditions for this concave optimization problem are sufficient, i.e. every $\delta^* \in \Delta(1)$ that satisfies them is a Nash product distribution. As shown in the proof of \Cref{thm:nashimp}, the KKT-conditions are given for every $x\in A$ with $\mu_x \geq 0$ by
	    \begin{align*}
	        \sum_{i \in N}\cnom_i \frac{u_i(x)}{u_i(\delta^*)}+\mu_x=1 \quad\text{and}\quad \left[\delta^*(x)>0  \text{ implies } \mu_x=0\right]\text.
	    \end{align*}
	    Assume that the dynamic procedure \emph{terminates}, i.e., for some $k$, $\delta^{k}(x)=\delta^{k+1}(x)=\delta^{k}(x) \sum_{ i \in N} \cnom_i\frac{u_i(x)} {u_i(\delta^{k})}$ for all $x \in A$. 
		 Recalling that $\delta^{k}(x) > 0$ for all projects $x\in A$ and $k\in \mathbb{N}$, $\delta^k$ satisfies the KKT-conditions and is a Nash product distribution.
	    
		In all other cases, let $\delta'$ be an accumulation point of $(\delta^k)_{k\in\mathbb N}$ and $(\delta^{k_l})_{l\in\mathbb N}$ be a subsequence converging to it.
		We show that $\delta'$ is a fixed-point of $f$. 
		The sequence $(F(f(\delta^{k_l}))-F(\delta^{k_l}))_{l\in\mathbb N}$ converges to $0$ by \Cref{step:nashdyn1}. 
		Continuity of $F$ implies $0=F(f(\delta^{'}))-F(\delta')\geq \frac{1}{2\log(2)}\norm{f(\delta')-\delta'}_1^2$,
		and so $f(\delta')=\delta'$.
		Therefore, $\delta'(x)=\delta'(x) \sum_{ i \in N} \cnom_i\frac{u_i(x)} {u_i(\delta')}$, which shows that $\delta'$ satisfies the KKT-conditions for all $x$ with $\delta'(x)>0$.
	    \\Denote by $S$ the set of all accumulation points. $S$ is connected as the step size of the dynamics converges to $0$ by \Cref{step:nashdyn1}. As $(F(\delta^k))_{k\in\mathbb N}$ converges, $F(\delta')=F(\delta^{''})$ for any two $\delta', \delta^{''} \in S$. If there exists a $\delta' \in S$ that has full support, then $\delta'$ and consequently, all accumulation points are Nash distributions as $(F(\delta^k))_{k\in\mathbb N}$ is increasing. 
	    \\In the remaining cases, every accumulation $\delta'$ point is located in a face $T_{\delta'}=\{\delta \in \Delta(1):\delta'(x)=0 \Rightarrow \delta(x)=0\}$ of $\Delta(1)$ and maximizes $F$ on this face by the fact that $\delta'$ has full support in $T_{\delta'}$. Therefore, $u_i(\delta')=u_i(\delta{''})$ for all $i \in N$ and $\delta', \delta^{''} \in T_{\delta'}$ and even for general $\delta', \delta^{''} \in S$ by connectivity of $S$. 
	    \\Assume now that there exist $\delta' \in S$ and $x \in A$ with $\delta'(x)=0$ but $\sum_{ i \in N} \cnom_i\frac{u_i(x)} {u_i(\delta')}>1$. This implies $\lim_{k \to \infty}\sum_{ i \in N} \cnom_i\frac{u_i(x)} {u_i(\delta^k)}>1$ which contradicts $\delta'(x)=0$.
	    
	\end{step}\setcounter{step}{0}
	
	\medskip
	Combining both steps, we conclude that every accumulation point of $(\delta^k)_{k\in\mathbb N}$ is a Nash product distribution and $(F(\delta^k))_{k\in\mathbb N}$ converges to the optimum Nash product as it is weakly increasing. 
	If the Nash product distribution is unique, $(\delta^k)_{k\in\mathbb N}$ thus has a unique accumulation point and converges (to the Nash product distribution).
\end{proof}
We mention some additional properties of this dynamic procedure. 
First, as noted by \citet{Cove84a}, one can bound the approximation error via $F(\delta^*)-F(\delta^k)\leq \max_{x \in A}\log\left(\sum_{ i \in N} \cnom_i\frac{u_i(x)} {u_i(\delta^k)}\right)$. Second, every distribution $\delta^k$ appearing in the sequence (apart from $\delta_0$) is decomposable, which is important when stopping after a finite number of steps. 
Finally, the procedure also converges to a Nash distribution in some cases where it is not unique. Suppose there are two `clone' projects $x$ and $y$ (such that all agents are indifferent between $x$ and $y$) but that the Nash distribution is unique if we were to merge these projects. Notice that if we start the dynamic procedure with the uniform distribution over all projects, then we have $\delta^k(x) = \delta^k(y)$ at each step $k$, which implies that the dynamic procedure does converge to a Nash distribution. 

\subsection{Contribution Incentive-Compatibility}
We now turn to our main result that \nash is contribution incentive-compatible.
The proof is technical and requires a number of lemmas, which are stated and proved in the appendix. 
At a high level, we estimate the rate of change of an agent's utility as her contribution increases, and integrate this quantity as she goes from not participating to participating in the mechanism to obtain the desired result. 
We are not aware of a simpler proof using the first-order conditions.
Attempts to prove \Cref{thm:nashir} by differentiating the first-order conditions with respect to $C_i$ (as in the proof of \Cref{thm:nashimp}) were unsuccessful.

\begin{restatable}{theorem}{nashir}\label{thm:nashir}
\nash is contribution incentive-compatible.
\end{restatable}

\begin{proof}
	Recall that we normalized utilities so that the utility assigned to least-preferred acceptable projects is 1 and so that the utility assigned to unacceptable projects is 0.
We must show that for all $\prof\in\profiles$ and $i\in N$,
\begin{align*}
	u_i(\nash(C_{-i}, C_i)) - C_i \ge u_i(\nash(C_{-i}, C'_i)) - C'_i \quad\text{for all $C'_i$ with $0 \le C_i' \le C_i$.}
\end{align*} 
Since \nash is invariant under replacing an agent with  utility function $u_i$ and contribution $C_i$ by two agents with utility function $u_i$ and contributions $C_i'$ and $C_i - C_i'$, respectively, it suffices to consider the case $C_i' = 0$.
Abusing notation, we write $C_{-i}$ for the profile with $(C_{-i})_i = 0$ and $(C_{-i})_j = C_j$ for $j\neq i$.
Consider the function $g\colon \profiles\rightarrow \Delta(1)$ with $g(\prof) = \nash(\prof) / \csum$ for all $\prof\in\profiles$.
We will show that
\begin{equation}\label{eq:irg}
	u_i(g(\prof)) \ge \frac{1}{\csum} ((\csum-C_i)u_i(g(\prof_{-i})) + C_i)\text,
\end{equation}
which is equivalent to the inequality above for $\nash$ with $C_i' = 0$.
We prove \eqref{eq:irg} with $i = 1$ as the focal agent.
For the remainder of the proof, fix the contributions $C_j$ of all agents $j \neq 1$, and assume that $C_j > 0$ for all $j \neq 1$. This is without loss of generality because $\nash$ ignores agents with zero contribution.

Denote by $\mathcal P_1\subseteq\mathbb R^n$ the polytope of feasible utility profiles scaled by $1/\csum$, i.e., $\mathcal P_1 = \{u(\delta)\colon \delta\in \Delta(1)\}$.
Since utility functions are linear, $\mathcal P_1$ is convex.
For $U\in \mathcal P_1$, let $F_\prof(U) = \sum_{i\in N}C_i \log U_i$.
Since by \Cref{lem:nashmax}, $F_\prof$ has a unique maximizer for all $\prof\in\profiles_{>0}$, we can define the function $\mathcal U\colon \profiles_{>0}\rightarrow \mathcal P_1$ that returns this unique maximizer for these profiles.
	
	Consider the function $\mathcal U_1(C_1) = u_1(g(C_1, C_{-1}))$ of agent 1's scaled utility as a function of $C_1$.
	If $\mathcal U_1(C_1) \ge 1$, then since $\mathcal U_1(C_1)$ is monotonically increasing in $C_1$ by \Cref{lem:continuous and monotone},
	\begin{align*}
		\mathcal U_1(C_1) = \frac1{\csum}\left((\csum - C_1)\mathcal U_1(C_1) + C_1\mathcal U_1(C_1)\right) \ge \frac1{\csum}\left((\csum - C_1)\mathcal U_1(0) + C_1\right),
	\end{align*}
	which proves~\eqref{eq:irg} in this case.
	The bulk of the proof is to derive a lower bound on the derivative of $\mathcal U_1(C_1)$ whenever $\mathcal U_1(C_1) < 1$.
	Then, integrating this derivative and using monotonicity of $\mathcal U_1$ gives \eqref{eq:irg}.
	
	\begin{step}\label{step:derivative}
 	Assume that $C_1 > 0$ and $\mathcal U_1(C_1) < 1$, and let $U = \mathcal U(\prof)$.
	Moreover, let $\mu\in (0,2)$ be arbitrary and let $\varepsilon^*$ be such that the conclusion of \Cref{lem:lintwicemax} holds;
	let $\varepsilon \in (0,\varepsilon^\ast)$. 
	Considering the Taylor expansion of the logarithm, there exists $\varepsilon' > 0$ such that for all $i\in N$ and $|r| < \varepsilon'$,
	\begin{equation}
		\left| \log(U_i + r) - \log U_i - \frac{r}{U_i} + \frac12 \left(\frac r {U_i}\right)^2\right| \le \frac{\varepsilon}{4} \left(\frac{r}{U_i}\right)^2\text.\label{eq:taylor}
	\end{equation}
	
	Now let $C'\in\profiles_{>0}$ be such that $C_1' = C_1 + dC_1$ with $0 < dC_1 <\min\{\varepsilon',\frac{\varepsilon}{(2 + \varepsilon)}\, C_1\}$ and $C_i' = C_i$ for all $i\in N\setminus\{1\}$.
	Consider the function $\phi\colon\mathbb R^n\rightarrow \mathbb R$ defined on $dU$ with $|dU|<\varepsilon^*$, such that
	\[
		\phi(dU) := F_{\prof'}(U + dU) - F_\prof(U) - dC_1\log U_1 = \sum_{i\in N} C_i \frac{dU_i}{U_i} + dC_1 \frac{dU_1}{U_1} - \psi(dU)\text,
	\]
	for some $\psi\colon\mathbb R^n\rightarrow \mathbb R$ with
	\[
		(1-\varepsilon)\frac12 \sum_{i\in N} C_i \left(\frac{dU_i}{U_i}\right)^2 \le \psi(dU) \le (1+\varepsilon)\frac12 \sum_{i\in N} C_i \left(\frac{dU_i}{U_i}\right)^2\text.
	\]	
	The existence of $\psi$ is guaranteed by \eqref{eq:taylor} and the bound on $dC_1$.
	
	Now let $U' = \mathcal U(\prof')$ and $dU' = U' - U$.
	Note that, since the only term in $\phi(dU)$ that depends on $dU$ is $F_{\prof'}(U + dU)$, $dU'$ maximizes $\phi$ among all $dU\in\mathbb R^n$ with $U + dU\in\mathcal P_1$.
	By \Cref{lem:polytope}, there is $\varepsilon''> 0$ such that, for all $dU\in\mathbb R^n$ with $|dU|\le\varepsilon''$ and $U + dU\in\mathcal P_1$, we have $U + rdU\in\mathcal P_1$ for all $r\in [0,2]$.
	Since $\mathcal U$ is continuous in $C$ by \Cref{lem:continuous and monotone}, $|dU'|$ will be small if $dC_1$ is small and we can choose $dC_1$ to be even smaller if necessary so that $2|dU'| \le \min(\varepsilon',\varepsilon'')$.
	Then, the function $\Phi\colon[0,2]\rightarrow \mathbb R$ with $\Phi(r) = \phi(rdU')$ is well-defined and satisfies the prerequisites of \Cref{lem:lintwicemax} with
	\[
	\alpha = \sum_{i\in N} C_i \frac{dU_i'}{U_i} + dC_1 \frac{dU_1'}{U_1} \quad\text{and}\quad \beta = \frac12 \sum_{i\in N} C_i \left(\frac{dU_i'}{U_i}\right)^2\text.
	\]
	Hence, it follows from \Cref{lem:lintwicemax} that
	\[
		\sum_{i\in N} C_i \frac{dU_i'}{U_i} + dC_1 \frac{dU_1'}{U_1} \ge\mu\Phi(1)\text.
	\]
	Since $U$ maximizes $F_\prof$, by \Cref{lem:first-order}, $\sum_{i\in N} C_i \frac{dU_i'}{U_i}\le 0$. It follows that
	\begin{equation}
	dC_1 \frac{dU_1'}{U_1} \ge\mu\Phi(1)\text.\label{eq:linmax}
	\end{equation}

	Next, let $\delta = g(\prof)$.
	Let $H_1 = \sum_{a\in A : u_1(a) > 0} \delta(a)$ be the fraction spent on agent 1's acceptable projects, i.e., those that agent 1 assigns positive utility.
	Recall that $U_1 < 1$, and so $H_1 < 1$.
	Since \nash gives agents with positive contribution positive utility, we have $H_1 > 0$. 
	From $0 < H_1 < 1$, we get that $\delta(a) < 1$ for all $a \in A$.
	Thus, for $|t|> 0$ small enough, take the distribution $\delta^t$ with 
	\[
	\delta^t(a) = 
	\begin{cases}
		(1+t)\delta(a) & \text{for all $a\in A$ with $u_1(a) > 0$,} \\
		(1-\frac{H_1}{1-H_1}\,t)\delta(a) & \text{for all $a\in A$ with $u_1(a) = 0$.}
	\end{cases}
	\]
	One can check that $\delta^t \in \Delta(1)$:
	\begin{align*}
	\sum_{a\in A} \delta^t(a)
	&= \sum_{\substack{a\in A\\u_1(a)>0}} (1+t)\delta(a) + \sum_{\substack{a\in A\\u_1(a)=0}} (1-\tfrac{H_1}{1-H_1}t)\delta(a) \\
	&= (1+t)\sum_{\substack{a\in A\\u_1(a)>0}} \delta(a) + (1-\tfrac{H_1}{1-H_1}t) \sum_{\substack{a\in A\\u_1(a)=0}} \delta(a) \\
	&= (1+t)H_1 + (1-\tfrac{H_1}{1-H_1}t) (1-H_1) = 1.
	\end{align*}
	Let $dU^t = u(\delta^t) - U$.
	For $|t|$ small enough, we have that $U + dU^t\in\mathcal P_1$ and $U - dU^t\in\mathcal P_1$. Indeed, $U + dU^t = u(\delta^t)$, and for the second statement we can perturb $\delta$ infinitesimally in the opposite direction. This is a valid perturbation because $\delta(a)<1$ for all $a\in A$, and for $a\in A$ such that $\delta(a)=0$ we have $\delta^t(a)=\delta(a)$.
	Thus, by \Cref{lem:first-order}, we have 
	\[
		\sum_{i\in N} C_i \frac{dU^t_i}{U_i} = 0\text.
	\]
	So for sufficiently small $|t|$, we have 
	\[
		\phi(dU^t) = dC_1 \frac{dU^t_1}{U_1} - \psi(dU^t) \ge dC_1 \frac{dU^t_1}{U_1} - (1+\varepsilon)\frac12 \sum_{i\in N} C_i \left(\frac{dU_i^t}{U_i}\right)^2\text.
	\]
	Since $dU_1^t = u_1(\delta^t) - U_1 = (1+t) U_1 - U_1$, we have that $\frac{dU_1^t}{U_1} = t$.
	Similarly, it follows that $-\frac{H_1}{1-H_1}\,t\le \frac{dU_i^t}{U_i} \le t$ for all $i\in N$. 
	
	Now, by definition of $H_1$, we have $U_1 \ge H_1$.
	Thus $1-U_1 \le 1 - H_1$. Hence $-\frac{U_1}{1-U_1}\le -\frac{H_1}{1-H_1}$. 
	Thus, applying \Cref{lem:auxinequality2} with $\alpha = \frac{U_1}{(1-U_1)}\,t$, $\beta = t$, and $x_i = \frac{dU_i^t}{U_i}$, it follows that
	\[	
		\phi(dU^t) \ge dC_1 t - (1 + \varepsilon) \frac12\frac{U_1 \csum}{1-U_1} t^2\text.
	\]
	Now let $t := \frac{1-U_1}{U_1 \csum}dC_1$.
	If $dC_1$ is small enough, then $t$ is also small enough and, recalling that $dU'$ maximizes $\phi$ among all $dU\in\mathbb R^n$ with $U + dU\in\mathcal P_1$, we get
	\[
		\Phi(1) = \phi(dU') \ge \phi(dU^t) \ge \frac12(1-\varepsilon)\frac{1-U_1}{U_1\csum}(dC_1)^2\text.
	\]
	Thus, by \eqref{eq:linmax}, we get
	\[
		dC_1\frac{dU_1'}{U_1} \ge \frac\mu2(1-\varepsilon)\frac{1-U_1}{U_1\csum}(dC_1)^2\text,
	\]
	from which it follows from $dC_1 > 0$ that
	\[
		dU_1' \ge \frac\mu2(1-\varepsilon)\frac{1-U_1}{\csum}dC_1\text.
	\]
	Since $\mu\in (0,2)$ was arbitrary and $\varepsilon > 0$ can be chosen arbitrarily small, it follows that 
	\[
		dU_1' \ge \frac{1-U_1}{\csum} dC_1\text.
	\]
\end{step}
	
\begin{step}\label{step:integral}
	We show \eqref{eq:irg} for $C_1 > 0$.
	(The case $C_1 = 0$ is trivial.)
	By \Cref{lem:continuous and monotone}, $\mathcal U_1(s)$ is monotonically increasing in $s\in[0,B_1]$.
	We have already proved \eqref{eq:irg} in the case $\mathcal U_1(C_1) \ge 1$. 
	Hence, we may assume $\mathcal U_1(s) < 1$ for all $s \in [0,C_1]$.
	
	Let $\epsilon\in (0,C_1)$ be arbitrary.
	By \Cref{step:derivative}, the lower right derivative of $\mathcal U_1$ at $s\in(\epsilon, C_1)$ is at least $\frac{1-\mathcal U_1(s)}{\csum - C_1 + s}$.
	Integrating this estimate from $\epsilon$ to $C_1$ yields
	\[
	-\int_{\epsilon}^{C_1} \frac{\frac{\partial\mathcal U_1(s)}{\partial s}}{1-\mathcal U_1(s)} ds \le -\int_\epsilon^{C_1} \frac{1}{\csum -C_1+s}ds
	\]
	from which we get
	\[
		\log(1-\mathcal U_1(C_1)) - \log(1-\mathcal U_1(\epsilon)) \le -(\log \csum - \log(\csum -C_1 + \epsilon))\text.
	\]
	Exponentiation yields $\frac{1-\mathcal U_1(C_1)}{1-\mathcal U_1(\epsilon)} \le \frac{\csum- C_1 + \epsilon}{\csum}$. 
	Since $\epsilon$ was arbitrary and $\mathcal U_1$ is monotonic, we get $\frac{1-\mathcal U_1(C_1)}{1-\mathcal U_1(0)} \le \frac{\csum- C_1}{\csum}$.
	Rewriting this equation gives us
	\[
		\mathcal U_1(C_1) \ge \frac1{\csum}\left((\csum - C_1)\mathcal U_1(0) + C_1\right),
	\]
	which is \eqref{eq:irg}.\qedhere
	\end{step}
\end{proof}

\section{Limits of Efficient Mechanisms}
\label{sec:limits}

In this section, we discuss the limits that we run into if we try to strengthen our notions of decomposability and contribution incentive-compatibility as described in Section~\ref{sec:prelims}. 
Specifically, we show that these strengthenings are incompatible with efficiency.

First, we consider \emph{strong decomposability}, which requires that $\delta$ can be divided into individual distributions $(\delta_i)_{i \in N}$ where for each $i \in N$, we have $\delta_i(x) > 0$ only if $i$ has assigned \emph{maximum} utility to $x$, i.e., only if $u_i(x) \ge u_i(y)$ for all $y \in A$.
In other words, each agent is only asked to spend her contributions on her favorite projects.

\begin{proposition}\label{thm:impleff}
No efficient mechanism satisfies strong decomposability when $m\ge 3$ and $n\ge2$.
\end{proposition}
\begin{proof}
\begin{table}[htb]
	\centering
	\[
	\begin{array}{rccc@{\hskip 3em}c}
		\toprule
		& u_i(a) & u_i(b) & u_i(x) & C_i\\\midrule
		\text{Agent 1} & 1+\varepsilon & 0 & 1 & 1\\
		\text{Agent 2} & 0 & 1+\varepsilon & 1 & 1\\
		\bottomrule
	\end{array}
	\]
	\caption[Type profile]{Profile  with $0 < \varepsilon < 1$ showing the incompatibility of strong decomposability and efficiency.}
	\label{tab:impeff}
\end{table}

	To see that strong decomposability is in conflict with efficiency, consider the example in \Cref{tab:impeff}.
	Here, both agents 1 and 2 have a pet project $a$ and $b$, respectively, which the other agent dislikes; there is also a compromise project $x$, which is close to optimal for both.
	It is best for an agent to spend her entire contribution on her pet project independently of what the other agent is doing.
	So the only allocation we can implement in the above sense is $a + b$, which gives utility $1+\varepsilon$ for both.
	But this fails to make use of the mutual interest in $x$:
	if they spent the whole pool of 2 on $x$, they could achieve utility 2 each.
\end{proof}

One can interpret the situation in \Cref{tab:impeff} as a prisoner's dilemma in which agents cooperate by spending on $x$ or defect (free-ride) by spending on $a$ and $b$.

The strengthening of contribution incentive-compatibility we discussed in Section~\ref{sec:prelims} especially requires that $u_i(f(\prof)) \ge u_i(f(\prof_{-i})) + C_iu_i^{\max}$, where $u_i^{\max} = \max_{y\in A}u_i(y)$. 
This \emph{strong contribution incentive-compatibility} cannot be satisfied in conjunction with efficiency.
The strong version makes sense if agents can use their money to fund public projects without going through the aggregation mechanism. This is typically the case for charities, but may be less applicable for some of the other scenarios discussed in \Cref{sec:intro}, such as residents of an apartment complex.

\begin{proposition}\label{thm:nondicho}
No efficient mechanism is strongly contribution incentive-compatible when $m\ge 4$ and $n\ge 3$.
\end{proposition}
\begin{proof}

\begin{table}[htb]
	\centering
	\[
	\begin{array}{rc@{\hskip 2em}c@{\hskip 2em}c@{\hskip 2em}c@{\hskip 3em}c}
		\toprule
		& u_i(a) & u_i(b) & u_i(c) & u_i(x) & C_i\\\midrule
		\text{Agent 1} & 2-\varepsilon & 0 & 0 & 1 & 1\\
		\text{Agent 2} & 0 & 2-\varepsilon & 0 & 1 & 1\\
		\text{Agent 3} & 0 & 0 & 2-\varepsilon & 1 & 1\\
		\bottomrule
	\end{array}
	\]
	\caption[Type profile]{Profile with $0 < \varepsilon < 0.5$ used in the proof of \Cref{thm:nondicho}.}
	\label{tab:nondicho}
\end{table}

Assume for contradiction that there exists a mechanism $f$ that is strongly contribution incentive-compatible and  efficient.
For $\prof = (1, 1, 1)$ as in \Cref{tab:nondicho},
the distribution $\delta = f(\prof)$ should only allocate resources to at most one of $a$, $b$, and $c$.
Otherwise, if there is any subset $\{y,z\} \subset \{a,b,c\}$, $y \neq z$ with $\delta(y)>0$ and $\delta(z)>0$, the distribution
\[
(\delta(y)-\kappa)\ y + (\delta(z)-\kappa)\ z + (\delta(x)+2\kappa)\ x
\] with $\kappa = \min(\delta(y),\delta(z))$ is strictly preferred by all three agents. Thus, without loss of generality, we can assume that  $\delta(c) = 0$.

Starting with agent $1$, we let the other agents join one after another and, using strong contribution incentive-compatibility, derive lower bounds on the resources allocated to project $x$.
It will turn out that after agent 3 has joined, the mechanism would have to allocate more than the whole pool of 3 to $x$ in order to accommodate strong contribution incentive-compatibility, which is a contradiction.

Let $\prof'= (1, 1, 0)$ and $\delta'= f(\prof')$. 
As above, efficiency implies that either $\delta'(a) = 0$ or $\delta'(b)=0$. Otherwise, if $\delta'(a)>0$ and $\delta'(b)>0$, the distribution
\[
(\delta'(a)-\kappa')\ a + (\delta'(b)-\kappa')\ b + (\delta'(x)+2\kappa')\ x 
\] 
with $\kappa' = \min(\delta'(a),\delta'(b))$ is strictly preferred by both agents with a utility improvement of $2\kappa'\varepsilon > 0$.
We assume that $\delta'(b) = 0$.
Treating the case $\delta'(a) = 0$ requires no more than switching the order of agents~1 and~2. 

By strong contribution incentive-compatibility, agent 2 must get at least the same utility as if both agents acted in an uncoordinated manner: $u_2(\delta') \geq u_2(a) + C_2 u_2^{\max} = u_2(a) + (2-\varepsilon) = 2-\varepsilon$ and with $\delta'(b)=0$, we have $\delta'(x)\geq \frac{2-\varepsilon}{u_2(x)} = 2-\varepsilon$.

Thus the utility of agent 3 from $\delta'$ can be bounded from below by $u_3(\delta') \geq \delta'(x)\ u_3(x) \geq (2-\varepsilon)$.

Applying strong contribution incentive-compatibility for agent 3 yields $u_3(\delta) \geq u_3(\delta') + C_3 u_3^{\max} \ge (2-\varepsilon) + 1 \cdot (2-\varepsilon) = 4-2\varepsilon$. As $\delta(c)=0$, agent 3 can only get positive utility from project $x$, and thus $\delta(x) = \frac{u_3(\delta)}{u_3(x)} \geq 4-2\varepsilon > 3$ for $0 < \varepsilon < 0.5$, which exceeds the pool of 3.
\end{proof}

The reason for this incompatibility is structurally similar to that for decomposability: efficiency requires spending resources on the compromise project $x$, but strong contribution incentive-compatibility can only be satisfied if the pet projects $a$, $b$, and $c$ are funded.

In computational experiments, it appears that $\nash$ satisfies a version of contribution incentive-compatibility that is stronger than the standard version but weaker than strong contribution incentive-compatibility. This version is inspired by the proportional spending property that we saw in the proof of \Cref{thm:nashimp} and in the dynamic procedure that converges to $\nash$ (\Cref{sec:computation}). Let $\prof \in \profiles$ be a contribution profile and write $\delta = \nash(\prof)$. For an agent $i \in N$ with $C_i > 0$ and for an amount $\epsilon > 0$ of potential extra contribution, let $\delta_\epsilon \in \Delta(\epsilon)$ be the distribution with $\delta_\epsilon(x) = \alpha \cdot \delta(x) \cdot u_i(x)$ for all $x \in A$, where $\alpha = \epsilon / (\csum \cdot u_i(\delta))$. Thus, $\delta_\epsilon(x)$ is proportional to the utility that $i$ derives from project $x$ in distribution $\delta$. We conjecture that $\nash$ satisfies
\[
u_i(\nash(C_{-i}, C_i + \epsilon)) \ge u_i(\delta) + u_i(\delta_\epsilon).
\]
This property lies between our two definitions of contribution incentive-compatibility because $\epsilon \le u_i(\delta_\epsilon) \le \epsilon \cdot u_i^{\max}$.

When only allowing dichotomous utility functions, both decomposability and contribution incentive-compatibility coincide with their strong counterparts. Hence, Propositions \ref{thm:impleff} and \ref{thm:nondicho} do not apply. In this restricted setting, which has been well-studied \citep{BMS05a,Dudd15a,ABM17a,BBPS21a}, \nash becomes an even stronger candidate mechanism, though Duddy's \citeyearpar{Dudd15a} \emph{conditional utilitarian rule}, which returns the decomposable distribution with the highest utilitarian welfare, constitutes an attractive alternative.

\subsection*{Acknowledgements}
This material is based on work supported by the Deutsche Forschungsgemeinschaft under grants {BR~5969/1-1}, {BR~2312/11-1}, and {BR~2312/12-1}, and by the European Research Council (ERC) under grant number 639945.
Preliminary results of this paper were presented at the AAMAS Workshop on Games, Agents, and Incentives (Montréal, May 2019).
We are grateful to Fedor Nazarov for suggesting the proof technique for \Cref{thm:nashir}, to Martin Bullinger for pointing out a more compact proof for \Cref{thm:nondicho}, and to Herv{\'e} Moulin, Klaus Nehring, and Erel Segal-Halevi for helpful and encouraging feedback.

\clearpage

\appendix

\section{Lemmas Used for the Proof of \Cref{thm:nashir}}

\begin{lemma}\label{lem:nashmax}
	For all $\prof\in\profiles_{>0}$, $F_\prof$ has a unique maximizer $U\in\mathcal P_1$.
	Moreover, $C_i \le U_i\le \csum u_i^{\max}$ for all $i\in N$.
\end{lemma}

\begin{proof}
	Assume for a contradiction that there are two distinct $U', U''\in\mathcal P_1$ which maximize $F_\prof$. 
	As a positive linear combination of strictly concave functions, $F_\prof$ is a strictly concave function. 
	Hence, for $U = \frac12\,(U' + U'')\in\mathcal P_1$, by strict concavity of $F_\prof$, we have
	\[
		F_\prof(U) > \frac12\left( F_\prof(U') + F_\prof(U'')\right) = F_\prof(U')\text,
	\]
	which contradicts the assumption that $U'$ maximizes $F_\prof$ over $\mathcal P_1$.	
	Now let $i\in N$.
	Clearly, $U_i$ is upper bounded by $\csum u_i^{\max}$.
	For the lower bound on $U_i$, recall from \Cref{thm:nashimp} that $\nash$ is decomposable.
	Thus the distribution $\nash(\prof)$ is the sum of distributions $\delta_j$, $j\in N$, such that $\delta_j\in\Delta(C_j)$ and $\delta_j(x) > 0$ only if $u_j(x) > 0$.
	In particular, $u_i(\delta_i) \ge C_i$, since we have normalized utility functions so that the lowest positive utility of each agent is 1.
\end{proof}
		
Recall that $\mathcal U\colon\profiles_{>0}\rightarrow \mathbb R_{\ge 0}^n$ returns the unique maximizer of $F_\prof$.
By \Cref{lem:nashmax}, $\mathcal U$ is well-defined.
We show that $\mathcal U(\prof)$ is continuous.
Moreover, the utility of every agent is weakly increasing in her contribution. 

\begin{lemma}\label{lem:continuous and monotone}
	$\mathcal U(C)$ is continuous in $C$ on $\profiles_{>0}$ and $\mathcal U_1(C_1)$ is weakly increasing in $C_1$.
\end{lemma}

\begin{proof}
	First we show that $\mathcal U$ is continuous in $C$ on $\profiles_{>0}$.
		Let $\prof\in\profiles_{>0}$ and consider a sequence $(C^k)_{k\in \mathbb N}\subseteq\mathbb R_{>0}^n$ converging to $C$.
		Further, let $U^k = \mathcal U(\prof^k)$ and $U = \mathcal U(\prof)$.
		Observe that since $C^k$ converges to $C$, by \Cref{lem:nashmax}, we have $0<\lambda\le U_i^k\le\Lambda$ for all $i$ and some $\lambda,\Lambda > 0$ and large enough $k$.
		Hence, by passing to a subsequence if necessary, we may assume that $U^k$ converges to $U^\ast$ for some $U^\ast\in\mathcal P_1$.
		Since the family of functions $F_\prof$, $F_{\prof^k}$, $k\in\mathbb N$, is uniformly equicontinuous on $[\lambda,\Lambda]^n$, it follows that $F_{\prof^k}(U^{k})$ converges to $F_\prof(U^\ast)$.
		Moreover, as $U^k$ maximizes $F_{\prof^k}$, we have $F_{\prof^k}(U^k) \ge F_{\prof^k}(U)$, which converges to $F_\prof(U)$.
		Hence, $U^\ast$ maximizes $F_{\prof}$, which, since $F_{\prof}$ has a unique maximizer by \Cref{lem:nashmax}, implies that $U^\ast = U$.
		Hence, $U^k$ converges to $U$.
	
	We prove that $\mathcal U_1(C_1) = \mathcal U(C_1,C_{-1})$ is weakly increasing in $C_1$.
	Let $s > 0$, $U = \mathcal U(C_1,C_{-1})$, and $U' = \mathcal U(C_1 + s,C_{-1})$.
	Assume for contradiction that $U_1' = \mathcal U_1(C_1 + s) < \mathcal U_1(C_1) = U_1$.
		Then,
		\[
			F_{\prof'}(U') = \sum_{i\in N} C_i \log U_i' + s \log U_1' < \sum_{i\in N} C_i \log U_i + s \log U_1 = F_{\prof'}(U)\text,
		\]
		where the inequality follows from the assumption that $U_1' < U_1$ and the fact that $U$ is a maximizer of $F_{\prof}$.
		This contradicts the assumption that $U'$ maximizes $F_{\prof'}$.
\end{proof}

\begin{lemma}\label{lem:polytope}
	For every $\prof\in\profiles$ and $U\in\mathcal P_1$, there is $\varepsilon > 0$ such that for all $dU\in\mathbb R^n$ with $|dU|\le \varepsilon$ and $U + dU\in\mathcal P_1$, we have $U + t dU\in\mathcal P_1$ for all $t\in[0,2]$.
\end{lemma}

\begin{proof}
	Since $\mathcal P_1$ is a polytope, it is an intersection of a finite number of closed half-spaces $H_i$. Observe that the desired property holds for each $H_i$. Indeed, if the point $U$ is in the interior of $H_i$, we can take $\varepsilon$ to be half of the distance from $U$ to the boundary of $H_i$, while if $U$ is on the boundary of $H_i$, the entire ray $\{U+tdU \mid t\geq 0\}$ is contained in $H_i$ and we can take $\varepsilon$ to be any positive real number. It follows that the desired property also holds for the intersection of the half-spaces $H_i$, which is $\mathcal P_1$.
\end{proof}

The next three lemmas will be useful for analyzing error terms obtained in the main analysis.

\begin{lemma}\label{lem:first-order}
	Let $\prof\in\profiles_{>0}$, $U = \mathcal U(\prof)$, and $dU\in\mathbb R^n$ such that $U + dU\in\mathcal P_1$.
	Then,
	\[
		\sum_{i\in N} C_i \frac{dU_i}{U_i}\le 0\text.
	\]
	If also $U - dU\in\mathcal P_1$, then equality holds.
\end{lemma}

\begin{proof}
		Consider the function $\tau\colon [0,1]\rightarrow \mathbb R$ with $\tau(t) = F_{\prof}(U + tdU)$ and observe that $\tau$ attains its maximum at $0$.
		Since $U_i> 0$ for all $i\in N$ by \Cref{lem:nashmax}, $\tau$ is differentiable at $0$.
		Hence, the right derivative of $\tau$ at $0$ is non-positive, i.e.,
		\[
			\frac{\partial \tau}{\partial t}\big |_{t = 0} = \frac{\partial}{\partial t}\left( \sum_{i\in N} C_i \log (U_i + t dU_i)\right)\big |_{t = 0} = \sum_{i\in N} C_i \frac{dU_i}{U_i} \le 0 \text.
		\]
		If additionally $U - dU\in\mathcal P_1$, the first part implies $-\sum_{i\in N} C_i\frac{dU_i}{U_i}\le 0$, from which equality follows.
\end{proof}

	\begin{lemma}\label{lem:auxinequality2}
		Let $\prof\in\profiles$, $x\in \mathbb R^n$, and $\alpha,\beta > 0$ such that $\sum_{i\in N} C_i x_i = 0$ and $-\alpha\le x_i \le \beta$ for all $i\in N$.
		Then,
		\[
			\sum_{i\in N} C_i x_i^2 \le \alpha\beta\sum_{i\in N} C_i\text.
		\]
	\end{lemma}
	
\begin{proof}
		Since $-\alpha\leq x_i\leq \beta$, we have $\left|x_i-\frac{\beta-\alpha}{2}\right|\leq \frac{\beta+\alpha}{2}$. It follows that
	\begin{align*}
	    \sum_{i\in N} C_ix_i^2 
	    &= \sum_{i\in N}C_i\left(x_i-\frac{\beta-\alpha}{2}\right)^2 - \left(\frac{\beta-\alpha}{2}\right)^2\sum_{i\in N}C_i \\
	    &\leq \left(\frac{\beta+\alpha}{2}\right)^2\sum_{i\in N}C_i - \left(\frac{\beta-\alpha}{2}\right)^2\sum_{i\in N}C_i \\
	    &= \alpha\beta\sum_{i\in N}C_i,
	\end{align*}
	as claimed.
\end{proof}

For the proof of \Cref{lem:lintwicemax}, we need the following auxiliary lemma.

	\begin{lemma}\label{lem:auxinequality}
		Let $\lambda^\ast \in (0,\frac12)$. 
		Then, there are $\varepsilon^\ast \in (0,1)$ and $t\in[1,2]$ such that
		\[
			t - \lambda \frac{1 + \varepsilon}{1-\varepsilon} t^2 > 1-\lambda \quad\text{for all $\lambda\in[0,\lambda^\ast]$ and $\varepsilon \in (0,\varepsilon^\ast)$.}
		\]	
	\end{lemma}
	
	\begin{proof}
	The inequality in the statement can be rewritten as 
$\lambda < \frac{t-1}{\frac{1+\varepsilon}{1-\varepsilon}t^2-1}.$
Choose an arbitrary $t\in (1,\frac{1}{\lambda^*} - 1)$. We have $t\in[1,2]$ and $\lambda^* < \frac{1}{1+t}$. Since $\lim_{\varepsilon\rightarrow 0}\frac{t-1}{\frac{1+\varepsilon}{1-\varepsilon}t^2-1}=\frac{1}{1+t}$, we can choose $\varepsilon^* \in (0,1)$ such that $\lambda^* < \frac{t-1}{\frac{1+\varepsilon}{1-\varepsilon}t^2-1}$ for all $\varepsilon\in(0,\varepsilon^*)$. It follows that $\lambda < \frac{t-1}{\frac{1+\varepsilon}{1-\varepsilon}t^2-1}$ for all $\lambda\in[0,\lambda^*]$ and $\varepsilon\in(0,\varepsilon^*)$, as desired.
	\end{proof}

	\begin{lemma}\label{lem:lintwicemax}
		For all $\mu\in (0,2)$ there is $\varepsilon^\ast \in (0,1)$ with the following property.
		For any $\Phi\colon [0,2]\rightarrow \mathbb R$ such that $\Phi(1) = \max_{t\in [0,2]} \Phi(t)$ and such that there are $\alpha,\beta\ge 0$ and $\varepsilon\in (0,\varepsilon^\ast)$ with
		\begin{align}
			\alpha t - (1+\varepsilon) \beta t^2 \le \Phi(t) \le \alpha t - (1-\varepsilon) \beta t^2 \label{eq:phiineq}
		\end{align}
		for all $t\in [0,2]$, it holds that $\alpha \ge \mu\Phi(1)$.
	\end{lemma}

\begin{proof}
	If $\mu\le 1$, then by choosing any $\varepsilon^* \in (0,1)$, we have $\mu\Phi(1)\le\Phi(1) \le \alpha$ by assumption.
	Assume henceforth that $\mu> 1$.
		Let $\lambda^*:=1-\frac{1}{\mu}>0$ and choose $\varepsilon^*>0$ and $t^*\in[1,2]$ such that 
	$$t^*-\lambda\frac{1+\varepsilon}{1-\varepsilon}(t^*)^2 > 1-\lambda$$ 
	for all $\lambda\in[0,\lambda^*]$ and $\varepsilon\in(0,\varepsilon^*)$, which is possible by \Cref{lem:auxinequality}.
	
	Let $\Phi$, $\alpha$, $\beta$, and $\epsilon$ as in the statement of the lemma.
	If $\alpha=0$, we get $\Phi(1)\leq 0$ by \Cref{eq:phiineq}.
	Hence, $\alpha\geq\mu\Phi(1)$ holds. 
	Now consider the case $\alpha>0$.
	Let $\lambda := \frac{\alpha - \Phi(1)}{\alpha} \geq 0$. 
	Assume for contradiction that the desired conclusion is not true, i.e., $\alpha < \mu\Phi(1)$. 
	This is equivalent to $\lambda < \lambda^*$. 
	Ihe function $\Psi(t):=\alpha t - \Phi(t)$ satisfies $\beta(1-\varepsilon)t^2 \leq \Psi(t) \leq \beta(1+\varepsilon)t^2$.
	By substituting $t=t^*$ and $t=1$, we have $\Psi(t^*)\leq \Psi(1)\frac{1+\varepsilon}{1-\varepsilon}(t^*)^2$. 
	It follows that
	\[
	    \Phi(t^*) 
	= \alpha t^* - \Psi(t^*) 
	\geq \alpha\left(t^* - \frac{\Psi(1)}{\alpha}\frac{1+\varepsilon}{1-\varepsilon}(t^*)^2\right) 
	= \alpha\left(t^* - \lambda\frac{1+\varepsilon}{1-\varepsilon}(t^*)^2\right) 
	> \alpha(1-\lambda) 
	= \Phi(1).
	\]
	This contradicts the assumption that $\Phi(1)=\max_{t\in[0,2]}\Phi(t)$.
\end{proof}

\section{Characterization of Decomposability}
\label{app:decomp}

Recall that for $i\in N$, $A_i = \{a\in A\colon u_i(a) > 0\}$ denotes the support of $u_i$.
Moreover, let $\bar A_i = \argmax\{u_i(a)\colon a\in A\}$.

\begin{proposition}\label{pro:decomposablecharacterization}
	Let $\delta\in\Delta(\csum)$ be a distribution.
	Then, $\delta$ is 
	\begin{enumerate}[label=(\roman*)]
		\item decomposable if and only if for every $N'\subseteq N$, $\sum_{x\in \bigcup_{i\in N'} A_i} \delta(x) \ge \sum_{i\in N'} C_i$;\label{item:decomp}
		\item strongly decomposable if and only if for every $N'\subseteq N$, $\sum_{x\in \bigcup_{i\in N'} \bar A_i} \delta(x) \ge \sum_{i\in N'} C_i$.\label{item:strongdecomp}
	\end{enumerate}
\end{proposition}

\begin{proof}
	We prove \ref{item:decomp}.
	It is easy to see that the inequalities hold if $\delta$ is decomposable.
	We prove the converse direction by an application of the strong duality theorem.
	A distribution $\delta$ is decomposable if and only if the following linear program $P$ has a solution with value $\csum$.\medskip\\
	\renewcommand{\arraystretch}{1.7}
	\begin{tabular}{cccccc}
	&primal ($P$)&&&&dual ($D$)\\
	&$\max \sum_{i \in N}\sum_{j \in A}x_{ij}$&&\quad&&$\min \sum_{i \in N}C_i y_i + \sum_{j \in A}\delta(j)y_j$\\
	s.t.&$\sum_{j \in A}x_{ij} \leq C_i$&$\forall i \in N$&&s.t.&$\forall (i,j) \in N \times A$:\\
	&$\sum_{j \notin A_i}x_{ij} \leq 0$&$\forall i \in N$&&&\multirow{2}{*}{$y_i+y_j \geq \begin{cases}
	1 & j \in A_i \\
	1-y_{n+i} & j \notin A_i
	\end{cases}$}\\
	&$\sum_{i \in N}x_{ij} \leq \delta(j)$&$\forall j \in A$\\
	&$x \geq 0$&&&&$y \geq 0$\\
	\end{tabular}
	\renewcommand{\arraystretch}{1}
	\medskip \\
	with $x=(x_{1a},x_{1b},\cdots,x_{2a},\cdots) \in \mathbb{R}_{\geq 0}^{n\cdot m}$ and $y=(y_1,\cdots,y_{2n},y_a,y_b,\cdots) \in \mathbb{R}_{\geq 0}^{2n+m}$. The second constraint corresponds to decomposability whereas $x_{ij}$ represents a possible contribution of agent $i$ on project $j$. Hence, $P$ has a solution with value $\csum$ if and only if $\delta$ can be decomposed according to $x$.\\
	Assuming $\sum_{j\in \bigcup_{i\in N'} A_i} \delta(j) \ge \sum_{i\in N'} C_i$ for every $N'\subseteq N$, we claim that there always exists a solution $y^*$ to its dual $D$ such that $y^*_{n+1}=\cdots=y^*_{2n}=0$. This means that we can reduce $D$ to $D'$ where the first constraint simplifies to $y_i+y_j \geq 1$ for all $(i,j) \in N \times A$. Looking at the dual of $D'$ called $P'$, we observe that compared to $P$, the constraint $\sum_{j \notin A_i}x_{ij} \leq 0$ for all $i \in N$ is removed and $P'$ has value $\csum$ as $\delta \in \Delta(\csum)$.\\
	As all of the stated problems have optimal solutions, the strong duality theorem implies that all four linear programs have the same optimal value $\csum$ and thus, $\delta$ is decomposable as then, $P$ has a solution with value $\csum$.
	\medskip \\
	To prove the claim, let $y$ be a solution to $D$ and $A_0$ the set of all $j \in A$ with $y_j=0$. This implies that for all $i$ with $A_i \cap A_0 \neq \emptyset$, $y_i \geq 1$ and we can set $y_{n+i}=0$. Denote the set of all such agents by $N_0$.\\
	Let $N'= N\setminus N_0$ and $j'=\argmin_{j\in \bigcup_{i\in N'} A_i} y_j$. Define $y_i'=y_i+y_{j'}$ for all $i \in N'$, $y'_j=y_j-y_{j'}$ for all $j\in \bigcup_{i\in N'} A_i$ and $y'_j=y_j$, otherwise. By construction, $y'$ is still feasible and 
	\begin{align*}
	&\sum_{i \in N}C_i y_i + \sum_{j \in A}\delta(j)y_j \\
	&= \sum_{i \in N\setminus N'}C_i y_i'+\sum_{i \in N'}C_i (y_i'-y_{j'})+\sum_{j \in \bigcup_{i\in N'} A_i}\delta(j)(y_j'+y_{j'})+\sum_{j \notin \bigcup_{i\in N'} A_i}\delta(j)y_j' \\
	&\geq \sum_{i \in N}C_i y_i'-\sum_{i \in N'}C_i y_{j'}+\sum_{i \in N'}C_iy_{j'}+\sum_{j \in A}\delta(j)y_j'=\sum_{i \in N}C_i y'_i + \sum_{j \in A}\delta(j)y'_j
	\end{align*}
	as $y_{j'}\sum_{j \in \bigcup_{i\in N'} A_i}\delta(j) \geq y_{j'}\sum_{i \in N'}C_i$ by assumption. \\
	Iterating this procedure with $y=y'$ until $N'=\emptyset$, we end in a solution $y^*$ to $D$ with $y^*_{n+1}=\cdots=y^*_{2n}=0$.
	\medskip \\
	The proof for \ref{item:strongdecomp} proceeds along the same lines. 
	The only difference is that now $\bar A_i$ is used instead of $A_i$.
\end{proof}

\Cref{pro:decomposablecharacterization} implies that, for the special case of dichotomous preferences, decomposability is equivalent to the \emph{fair group share} axiom introduced by \citet{BMS02a} and later dubbed \emph{proportional sharing} by \citet{Dudd15a}.

\section{Independence of Contribution Incentive-Compatibility and Decomposability}
\label{app:contributionnotdecomposable}

We first define a mechanism that satisfies decomposability but violates contribution incentive-compatibility. To this end, we consider a rule that always returns the decomposable distribution with \emph{minimal} utilitarian welfare and thus represents an antipode to the conditional utilitarian rule $\cut$ introduced by \citet{Dudd15a},  Let $A_i^{\text{min}}=\{a \in A\colon u_i(a) > 0 \text{ and } \sum_{x \in N} u_x(a) \leq \sum_{x \in N} u_x(b) \text{ for all } b\text{ with } u_i(b)>0\}$. Then, 
\[
\anticut(\prof)=\sum_{i \in N} \sum_{x \in A_i^{\text{min}}} \frac{C_i}{|A_i^{\text{min}}|} \cdot x\text.
\]
This mechanism is decomposable by construction but fails to satisfy contribution incentive-compatibility.
For example, let $N = \{1,2\}$, $A = \{a,b\}$.
Let $u_1 = 1_{\{a,b\}}$ and $u_2 = 1_{\{a\}}$, and $B_1 = B_2 = 1$.\footnote{For $B\subset A$, denote by $1_B$ the dichotomous utility function with $1_B(x) = 1$ for $x \in B$ and $1_B(x) = 0$ for $x \not\in B$.}
Then, $\anticut((1,1))= a + b$ and $\anticut((1,0))=0.5\cdot a + 0.5\cdot b$.
Since 
\[
	u_2(0.5\cdot a + 0.5\cdot b) - 0 = 0.5 > 0 = u_2(a + b) - 1,
\]
\anticut violates contribution incentive-compatibility.
\medskip\\
Second, we construct a mechanism that is contribution incentive-compatible but violates decomposability for 
$N = \{1,2,3\}$ and $A = \{a,b,c,d\}$.
(This mechanism can be straightforwardly extended to more agents and projects.)
Let $u_1 = 1_{\{a,b\}}$, $u_2 = 1_{\{a,c\}}$, and $u_3 = 1_{\{d\}}$, and $B_1 = B_2 = B_3 = 1$.

We define $f(\prof) = \min\{C_1,C_2\}\cdot a + (C_1 - \min\{C_1,C_2\})\cdot b + (C_2 - \min\{C_1,C_2\})\cdot c + (\min\{C_1,C_2\} + C_3)\cdot d$.
Then, $f$ is not decomposable since, for example, if $C_1 = C_2 = C_3 = 1$, $f(\prof) = a + 2\cdot d$, which is not decomposable. 
On the other hand, one can (by distinguishing the cases $C_1 < C_2$ and $C_1 \ge C_2$) verify that $f$ is contribution incentive-compatible.

\end{document}